\documentclass{article}

\usepackage{arxiv}

\usepackage[utf8]{inputenc} 
\usepackage[T1]{fontenc}    
\usepackage{hyperref}       
\usepackage{url}            
\usepackage{booktabs}       
\usepackage{amsfonts}       
\usepackage{nicefrac}       
\usepackage{microtype}      
\usepackage{lipsum}		
\usepackage{graphicx}
\usepackage{natbib}
\usepackage{doi}

\usepackage{amsmath}
\usepackage{amssymb}
\usepackage{graphicx}
\graphicspath{{./images/} }
\usepackage{multirow}
\usepackage{xcolor}
\usepackage[all]{xy}
\usepackage{csquotes}

\usepackage{float}

\usepackage[T1]{fontenc}
\usepackage{natbib}

\usepackage{enumitem}
\usepackage{bm}
\usepackage{bbm}

\usepackage{tikz}
\usetikzlibrary{positioning}
\usepackage{subcaption}

\usepackage{multirow}
\usepackage{graphicx}
\usepackage{dsfont}
\usepackage{url}
\usepackage[]{appendix} 

\usepackage{amsthm}
\newtheorem{theorem}{Theorem}
\newtheorem{proposition}{Proposition}
\newtheorem{definition}{Definition}

\newcommand{\ci}{\mbox{\protect $\: \perp \hspace{-2.3ex}\perp$ }}
\newcommand{\nci}{\mbox{\protect $\: \perp \hspace{-2.3ex}\perp $} \hspace{-2mm}/ \hspace{1.5mm}}

\title{\normalfont Causal invariance in graphical models with latent variables}

\date{} 					

\author{ Marco Borriero \\
	Department of Statistics, Computer Science \\ and Applications,\\
	University of Florence,\\
    Viale Giovanni Battista Morgagni 59, 50134 \\ Florence, Italy,\\
	\texttt{marco.borriero@unifi.it} \\
	\And
	  Monia Lupparelli \\
	Department of Statistics, Computer Science \\ and Applications,\\
	University of Florence,\\
    Viale Giovanni Battista Morgagni 59, 50134 \\ Florence, Italy,\\
	\texttt{monia.lupparelli@unifi.it} \\
	\AND
	Giovanni M. Marchetti \\
	Department of Statistics, Computer Science \\ and Applications,\\
	University of Florence,\\
    Viale Giovanni Battista Morgagni 59, 50134 \\ Florence, Italy,\\
	\texttt{giovanni.marchetti@unifi.it} \\
	\And
	Veronica Vinciotti \\
	Department of Mathematics, \\
    University of Trento,\\
    Via Sommarive 14, 38123 \\ Povo (TN), Italy \\
    \texttt{veronica.vinciotti@unitn.it}
}

\renewcommand{\headeright}{}
\renewcommand{\undertitle}{}
\renewcommand{\shorttitle}{Causal invariance in graphical models with latent variables}

\hypersetup{
pdftitle={A template for the arxiv style},
pdfsubject={q-bio.NC, q-bio.QM},
pdfauthor={David S.~Hippocampus, Elias D.~Striatum},
pdfkeywords={First keyword, Second keyword, More},
}

\begin{document}
\maketitle

\begin{abstract} 
Causal discovery aims to identify causal relationships among variables  from observational or interventional data,
typically represented by a  directed acyclic graph (DAG). The causal invariance principle enables the identification  of the causal parents of target variables by exploiting the stability of causal effects across different experimental settings. When some  parents are unobserved, however, the induced  graph  over the observed  variables may no longer be a DAG, and it may not be unique, complicating causal inference.  
For relevant configurations of latent parents, we characterize  the induced graph and formalize the conditions 
under which  causal invariance is preserved for the identification of the observed parents. 
Necessary and sufficient conditions for testing such invariance are formally established for a multivariate Gaussian target.
\end{abstract}


\section{Introduction}

Causal directed acyclic graphs (DAGs) are extensively used to represent the structural causal model of data generating processes. In these graphs, the nodes correspond to random variables and causal relationships are encoded by directed edges, where $X \rightarrow Y$ implies that $X$ is  a causal parent of $Y$ \citep{pearl2009causality}.
Identifying the set of causal parents $X_{PA}$ of a target  $Y$ requires an inferential procedure known as causal discovery. 
Invariant Causal Prediction (ICP), introduced by \citet{peters}, is a recent method for parent search  based on the idea that causal effects remain invariant 
across different experimental settings,  referred to as environments.  Given data sampled from these settings,
the method  consists in testing the invariance of the conditional distribution of $Y$ given $X_S$ across environments for all  subsets $X_S$ of the predictor set $X_V=\{X_i\}_{i=1}^p$. When all causal parents are observed,  the intersection of the subsets  $X_S$ that satisfy the invariance hypothesis belongs to  $X_{PA}$ with controllable probability. 

In the absence of causal sufficiency, that is when some parents   are unobserved in the generating DAG, testing causal invariance  may fail because  hidden  variables induce distortions, confounding, and  identifiability issues. 
Related approaches, such as causal Dantzig  \citep{rothenhausler2019causal}, provide closed-form estimators of causal effects under the more restrictive setting of structural equation models with additive perturbations and allow correlated independent errors  for potential hidden variables.  \citet{long2023estimating}  extend this method by introducing  efficient hybrid estimators that make use of instrumental variables for model identifiability. Alternative approaches focus on robust prediction methods under distributional shifts \citep{gnecco2026boosted,henzi25,rothenhausler2021anchor}.
The impact of latent structures in testing causal invariance for parent identification requires further investigation as pointed out  recently by \citet{gnecco2026boosted}. 

We approach this issue by investigating whether  the    identification of the  parent set can be  recovered by the graphical  model induced over the observed variables. More specifically, the general question behind this paper is  whether the   graph  induced over the observed variables  enables one to specify the structural causal model of a target $Y$  after integrating out the latent variables.  This is an ambitious question as characterizing the    
graphical causal model for the observed  variables is far from straightforward. 
The induced graph may no longer be  a DAG and mixed graphs with further bi-directed edges are often  needed to encode the intricate  dependence structures arising   after marginalizing over the latent variables \citep{SadLau-2014}.
Additionally, depending on the graphical marginalization criterion adopted \citep{Zhang, richardson2003}, the induced graph and the resulting parent sets for $Y$ may not be unique, thereby complicating 
 causal parent identification.
 
In this paper, we provide an answer to this question under certain latent configurations. In particular, we consider closely the cases when unobserved parents  act either as a \emph{transition node}, $X \rightarrow H \rightarrow Y$, or as a \emph{source node}, $X \leftarrow H \rightarrow Y$ along the causal path, which are the cases of hidden mediator 
 and hidden counfounder, respectively (first row of Figure \ref{fig:DAG-latent}). We first characterize  the  graphical model for the observed variables such that causal invariance is preserved  for the distribution of a target of interest conditional on the observed parents. Then, we develop 
a framework for testing such invariance when the target is a multivariate Gaussian variable; the univariate Gaussian target represents a special case. Empirical examples illustrate how  the method can recover causal parents under these latent structures. Extended proofs and technical details  are deferred to the Appendix.

\begin{figure}[t]
\begin{minipage}{0.22\textwidth}
\centering

\[
\xymatrix{
& H \ar[dl]   & Z \ar[d]\\
Y  & & X \ar[ul] \\}
\]
(a)
\end{minipage}
\hfill
\begin{minipage}{0.22\textwidth}
\centering

\[
\xymatrix{
& H \ar[dl]   &Z \ar[d]\\
Y  & & X \ar[ul] \ar[ll] \\}
\]
(b)
\end{minipage}
\hfill
\begin{minipage}{0.22\textwidth}
\centering

\[
\xymatrix{
& H \ar[dl] \ar[dr]  &Z \ar[d]\\
Y  & & X  \\}
\]
(c)
\end{minipage}
\hfill
\begin{minipage}{0.22\textwidth}
\centering

\[
\xymatrix{
& H \ar[dl] \ar[dr]  &Z \ar[d]\\
Y  & & X \ar[ll] \\}
\]
(d)
\end{minipage}
\begin{minipage}{0.22\textwidth}
\centering
\vspace{0.3cm}
\[
\xymatrix{
Y & X \ar[l] & Z \ar[l]
}
\]
(e)
\end{minipage}
\hfill
\begin{minipage}{0.22\textwidth}
\centering
\vspace{0.3cm}
\[
\xymatrix{
Y & X \ar[l] & Z \ar[l]
}
\]
(f)
\end{minipage}
\hfill
\begin{minipage}{0.22\textwidth}
\centering
\vspace{0.3cm}
\[
\xymatrix{
Y \ar@{<->}[r] & X & Z \ar[l]
}
\]
(g)
\end{minipage}
\hfill
\begin{minipage}{0.22\textwidth}
\centering
\vspace{0.3cm}
\[
\xymatrix{
Y & X \ar[l] \ar@/^1pc/@{<->}[l] & Z \ar[l]
}
\]
(h)
\end{minipage}
\caption{First row: DAGs with latent parent  configurations in a simple setting of three observed variables $Y,X,Z$; latent variable $H$ is   a transition node in graphs (a)-(b) and a source node in graphs (c)-(d). Second row: corresponding induced mixed graphs over the observed variables.  }
\label{fig:DAG-latent}
\end{figure}
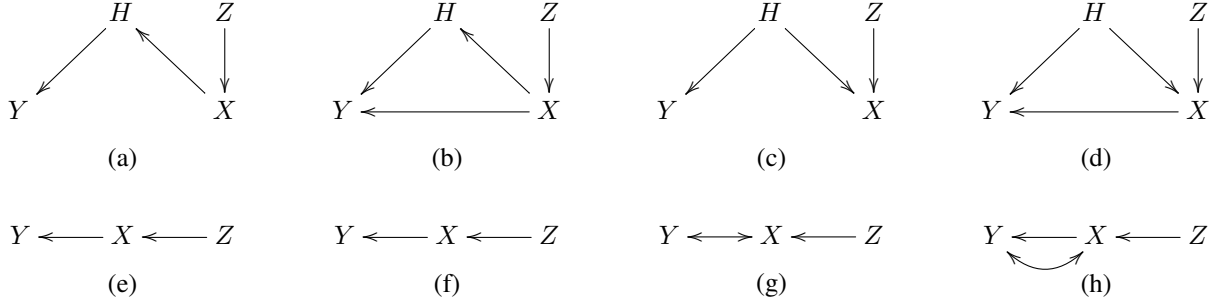

\section{Preliminaries}\label{sec.preliminary}
\subsection{Invariance condition of causal model}

Let us consider a DAG $\mathcal{D}=(V,E)$, representing the generating process for a set $X_V=\{X_i\}_{i \in V}$ of  random variables, where each  $X_i$ is associated to a node of the graph and $E$ is the set of directed edges. 
A  DAG  model 
is an independence model for    $X_V$ where conditional independences are implied by $d$-separation in the associated graph.
A causal DAG model, where the arrows imply causal relationships, is defined by an 
additional assumption of invariance, or stability, of the conditional distribution of each node $X_i$ given its causal parents $X_{PA_i}$ under an intervention on all or a subset of $X_{PA_i}$ \citep{pearl2009causality}. If we represent the model  via the structural causal model 
\begin{equation}\label{SEM}
X_i = f_i(X_{PA_i}, \delta_i), \quad  \forall i \in V,
\end{equation}
where $f_i$ is a deterministic function, $X_{PA_i}$ are the causal parents of $X_i$ in the graph 
and $\delta_i$ are independent random errors,  causal invariance refers to the invariance of $f_i$ and of the distribution of $\delta_i$ under intervention.  
We concentrate here on a target variable $Y \in X_V$ and on the identification of its causal parents   
$X_{PA}$. 
The invariance principle implies that if an external intervention is applied on, say, all causal parents  of $Y$, by setting these to values $x_{PA}$,  then
\begin{equation}
p(Y \vert \text{do}(X_{PA}=x_{PA}))=p(Y \vert X_{PA}=x_{PA}),
\end{equation}
that is, the distribution of $Y$ post-intervention, denoted using the \text{do}-operator, is the same as the conditional distribution of $Y$ when $X_{PA}$ are simply observed to take the value $x_{PA}$. %
The same definition can be given for  a  multivariate target $Y=\{Y_j\}_{j=1}^m$ where $X_{PA}= \cup_{j=1}^m X_{PA_j}$ and $X_{PA_j}$ is the parent set of $Y_j$.

Since external interventions are often unethical or unfeasible,  \citet{peters} and related approaches for general response types \citep{kook25,alice} exploit the heterogeneity of datasets across different environments $e \in \mathcal{E}$ to mimic  the effect of an external distributional intervention on $X_V$ without actually performing it. 
Under the assumption of no intervention on $Y$, they provide an inferential procedure for parent identification based on testing the invariance condition using the data collected in different environments. 

\begin{definition} \label{Invariance definition}
Consider a set of data $(Y^e, X^e)$  collected from different experimental settings $e \in \mathcal{E}$, where $Y^e$ is a (multivariate) target variable which is not subject to interventions and $X^e$ represents a set of predictors. Let $Y^e = f(X_{PA}^e, \delta_Y^e)$ be the structural equation for the target for each   $e \in \mathcal{E}$. The conditional distribution $p(Y \vert X_{PA})$ is invariant if: (i) $f$ is the same $\forall e \in  \mathcal{E}$, (ii) the distribution of $\delta_Y^e$ is the same for each $e \in \mathcal{E}$ and (iii) $\delta_Y^e \ci X_{PA}^e$.
\end{definition}
Intuitively, testing invariance  corresponds to verifying that the conditional distributions $Y^e \vert X_{PA}^e=x_{PA}^e$ and $Y^{e'} \vert X_{PA}^{e'}=x_{PA}^{e'}$ are equal for all environments $e,{e'} \in \mathcal{E}$, which mimic an external intervention at least on the parent set in order to ensure their correct  identification.

\subsection{ Causal model with latent  variables}
We now consider causal DAGs    in which
some of the causal parents of the target variables $Y$ are unobserved. When adequate a priori information is lacking and imposing probabilistic assumptions 
on the latent structures is too restrictive, working with the marginal distribution of the observed variables becomes a desirable alternative.
The graph for the observed variables  may no longer be a DAG, since this class is not closed under marginalization. A  larger class of  graphs is then needed. This class includes additional bi-directed edges to represent the  dependence induced between two observed variables by latent common causes. 

In particular, we consider the class of acyclic directed mixed graphs \citep{richardson2003} (henceforth, mixed graphs)  in which  nodes may  be linked by multiple directed and bi-directed edges that capture the effect of latent structures. Such graphs and related models have been employed to estimate causal effects in the presence of latent variables in \citet{henckel-al-2024}. 
As an example, the second row of Figure \ref{fig:DAG-latent} displays the  mixed graphs induced by  the  latent parent configurations  in the first row. When a transition variable $H$ is ignored (Figures \ref{fig:DAG-latent}(a) and (b)), as typical in mediation analysis, the induced mixed graph remains a DAG, since ignoring a transition node yields a directed edge between the observed variables. Note that $X$ becomes a parent of $Y$ in graph (e) even though it is an ancestor in the generating DAG (a). When the latent  parent is a source node (Figures \ref{fig:DAG-latent}(c) and (d)), the induced mixed graph is not a DAG.  In particular, ignoring $H$ in DAG (c) yields no directed relationship between  $Y$ and $X$ in the induced  graph (g), while in the mixed graph (f) the multiple edges between $Y$ and $X$ enable one to account for the confounding effect, distinguishing   the direct  effect of $X$ from the non-direct dependence induced by $H$ in DAG (d). 

Graphical mixed models are ruled by the $m$-separation criterion that generalizes the $d$-separation criterion for mixed graphs \citep{richardson2003}. The resulting model is a structural equation model with dependent errors and further dependences stemming from latent variables. Let $\mathcal{G}=(V, E)$ be a mixed graph associated with a random vector $X_V$, and let the model in \eqref{SEM} hold for any $X_i \in X_V$. Unlike DAGs, $\delta_i \nci \delta_j$  if  $X_i$ and $X_j$ are linked by  a bi-directed edge ($X_i \leftrightarrow  X_j$). Moreover, if $i \in PA_j$ such that the nodes are linked by multiple edges  $(X_i \mathrel{\substack{\longleftrightarrow \\[-0.5ex] \longrightarrow}} X_j)$, then $\delta_j \nci X_{PA_j}$. For instance, the structural equation model associated to the mixed graph in Figure \ref{fig:DAG-latent}(h) is
\begin{eqnarray*}
    Y = f_Y(X, \delta_Y), \quad X = f_X(Z, \delta_X), \quad Z =\delta_Z,  \qquad \delta_Y \nci \delta_X, \; \; \delta_Y \nci X.
\end{eqnarray*}
Instead, the model associated to the mixed graph in Figure \ref{fig:DAG-latent}(g) is defined by the equations
\begin{eqnarray*}
    Y =  \delta_Y, \quad X = f_X(Z, \delta_X), \quad Z =\delta_Z,  \qquad \delta_Y \nci \delta_X.
\end{eqnarray*}
The implications that marginalization over the latent variables has on causal invariance and, therefore, whether  causal discovery can still be conducted on the statistical model associated to the induced mixed graph are central questions to this paper.
In particular, given a causal DAG with  latent parents, we characterize the structural causal model for the target variable in the induced mixed graph and we formalize the testing of causal invariance for the identification of the observed parents when the target is Gaussian. Causal graphical models without and with hidden confounders are considered, respectively,  in Sections \ref{sec.SCM} and \ref{sec.SCMH}.

\section{Structural causal models without hidden confounders}\label{sec.SCM}
\subsection{Mixed graph model for the observed variables}
Let $X\in \mathbb{R}^p$ be a set of predictors, $H \in \mathbb{R}^q$ a set of latent variables and $Y \in \mathbb{R}^m$ a multivariate response whose parent sets of observed and unobserved variables are ${X}_{PA}$ and $H_{PA}$, respectively. The parents of $H_{PA}$ are denoted with $X_{PA(H)}$ instead.
Based on a data generating DAG mechanism, let the model of $Y$ and of the latent parents $H_{PA}$ follow the causal equations
\begin{equation}\label{eq:SCM-transition-source}
Y = f(X_{PA}, H_{PA}, \delta_Y), \quad  H_{PA} = g(X_{PA(H)}, \delta_H), \qquad \delta_{Y} \ci \{{X}_{PA}, H_{PA}\}, \; \;\; \delta_H \ci \delta_Y,
\end{equation}
with $\delta_{Y} \sim F_{\delta_{Y}}$
for some distribution $F_{\delta_{Y}}$.  
The latent variables in $H_{PA}$ act as transition and/or source node along the paths to $Y$ without confounding effects, such that $\varepsilon_Y \ci X_{PA}$. Interventions do not affect the unobserved variables, as well as the target variables. Therefore, the distributions of both $Y|X_{PA}$ and $H_{PA}|X_{PA(H)}$ are invariant, since (\ref{eq:SCM-transition-source}) is a structural causal model.

Marginalizing over $H$ leads to a mixed graph 
for the  observed variables where, ignoring a transition latent parent $H_l$ along the causal path to $Y_j$, induces a directed edge such that the  parents of $H_l$  become parents of  $Y_j$, while ignoring a source latent parent $H_l$ induces bi-directed edges linking its child nodes. Figures \ref{fig:DAG-latent}(a)-(c) and Figures \ref{fig:DAG-latent}(e)-(g) are illustrative of the marginalization criterion for a single target. 
No multiple edges arise in the induced graph since there are no confounding effects. Under these latent configurations, the induced  mixed graph is equivalent to an ancestral graph \citep{Zhang}, which is another notable class of graphs used to represent the marginal distribution of DAGs with hidden variables, using a different marginalization criterion.

The  model for  $Y$ given the observed variables corresponds to the structural equation model 
\begin{equation}\label{structural causal model without HC}
Y =\tilde{f}(X_{\widetilde{PA}},  \varepsilon_{Y}), \qquad \varepsilon_{Y} \ci {X}_{\widetilde{PA}},
\end{equation}
for an error term $\varepsilon_Y \sim F_{\varepsilon_y}$ and some function $\tilde{f}$ where $X_{\widetilde{PA}}=X_{PA} \cup X_{PA(H)}$ is the augmented parent set. The next proposition proves that the marginal model in \eqref{structural causal model without HC} is a structural causal model for $Y$
and the marginal mixed graph can be given a causal interpretation for the target $Y$.
\begin{proposition}\label{Proposition 1}
 The marginal model in \eqref{structural causal model without HC}  is a  structural causal model, then the conditional distribution of $Y \vert X_{\widetilde{PA}}$ is invariant. 
\end{proposition}
The above result proves that ignoring  source variables does not induce causal relationships among the observed variables. On the other hand, when mediator/transition variables are not observed, the stability of the causal mechanism that links the augmented parents $X_{\widetilde{PA}}$ to the target $Y$ is preserved as no intervention is made on the latent variables.
In case of hidden confounding  such that $X_{PA} \nci \delta_H$ in \eqref{eq:SCM-transition-source}, Proposition \ref{Proposition 1} no longer holds, as will be discussed in detail in Section \ref{sec.SCMH}.

\subsection{Gaussian target variable}

Given the   mixed graph  and the related structural causal model \eqref{structural causal model without HC} for $Y$, the next Theorem provides a necessary and sufficient condition for testing causal invariance for any subset $X_S \subseteq X$  in a linear model with Gaussian target variables $Y$.

\begin{theorem}\label{Teorema 1}
Let $Y =B^T_{{PA}}X +  \varepsilon_{Y}$ be a linear structural equation for the target $Y$, for some coefficient matrix $B_{{PA}} \in \mathbb{R}^{p \times m}$
where $(B_{{PA}})_{ij}=0$ if and only if $X_i \not \in X_{\widetilde{PA}_j}$ with $X_{\widetilde{PA}_j}$ the set of observed parent variables of $Y_j$ in the induced mixed graph. Assume
$\varepsilon_{Y} \sim N(0, \Sigma)$ for some covariance matrix $\Sigma$ and $\varepsilon_{Y} \ci X_{\widetilde{PA}}$. Let $S_j \subset \{1,\dots,p\}$ be a set of potential causal parents for $Y_j \in Y$ with $S=\bigcup_{j=1}^{m}S_j$. Let $B \in \mathbb{R}^{p \times m}$ such that $B_{ij}=0$ if and only if $X_i \not \in S_j$ and that solves the expected likelihood score equations
\begin{equation}\label{ortogonalità}
\mathbb{E}_{{X}_{S},{Y}}[({Y} - B^T{X}){X}_{S}^T]=0.
\end{equation}
Then, $S=\widetilde{PA}$ and $B=B_{PA}$ almost surely if and only if $\mathbb{E}_{X,Y}[(Y - B^T X)(Y - B^T X)^T] = \Sigma$. 
\end{theorem}
Theorem \ref{Teorema 1} can be exploited in order to build an inferential procedure for recovering the parent set of a given target variable. 
The proposed procedure considers all possible subsets  $X_S \subseteq X$, and for each of these, tests the invariance condition 
\begin{equation} \label{H0}
H_0 : \Sigma_1=\dots=\Sigma_k
\end{equation}
from data across the $k$ environments $e \in \mathcal{E}$.
To give an empirical illustration of this procedure, we employ Box's M test using the test statistic 
\[M = (n-k)\log |\hat{\Sigma}| - \sum_{e=1}^{k}(n_e-1) \log  |\hat{\Sigma}_e|,\]
with $\hat{\Sigma}^e$ the unbiased $m\times m$ sample covariance  matrix of the residuals of a multivariate linear regression model predicting the response using data from environment $e$ and $\hat{\Sigma}= \frac{1}{n-k}\sum_{e=1}^{k}(n_e - 1) \hat{\Sigma}_e$ the pooled covariance, with $n_e$ denoting the sample size of environment $e$ and $n=n_1+\dots+n_k$ the total sample size. Setting $c=\frac{2m^2 + 3m -1}{6(m+1)(k-1)}\left(\sum_{e=1}^{k}\frac{1}{n_e -1}- \frac{1}{n-k} \right)$, under the null hypothesis, 
\[M(1-c) \overset{H_0}{\sim} \chi^2_{\nu}, \quad \nu = (k-1)m(m+1)/2.\] 
At the significant level $\alpha$, we expect the invariance condition \eqref{H0} to be accepted for more than one subset $S$, as  non-predictive variables can be included in any invariant set while preserving invariance. For this reason, among all sets for which the null hypothesis cannot be rejected, we select the one that minimizes the  average Bayesian Information Criterion (BIC)  across environments. %

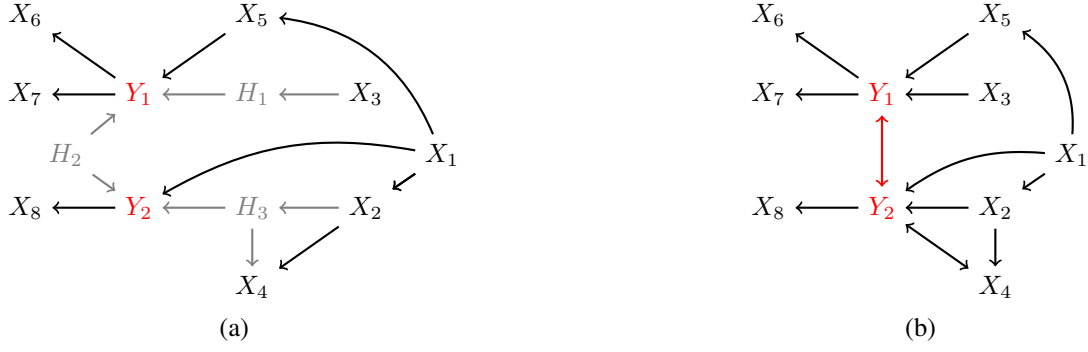
\begin{figure}
\begin{minipage}{0.45\textwidth}
\centering 
\begin{tikzpicture}[node distance={15mm}, thick, main/.style = {}] 
\node[main, text=red] (y1) {$Y_1$}; 
\node[main, text=red] (y2) [below of=y1]{$Y_2$}; 
\node[main, text=gray] (h2) [below left = 0.25cm and 0.3cm of y1] {$H_2$};
\draw[->, gray] (h2) -- (y1);
\draw[->, gray] (h2) -- (y2);
\node[main, text=gray] (h1) [right of=y1]{$H_1$};
\draw[->, gray] (h1) -- (y1);
\node[main] (x3) [right of=h1]{$X_3$};
\draw[->, gray] (x3) -- (h1);
\node[main, text=gray] (h3) [below of=h1]{$H_3$};
\draw[->, gray] (h3) -- (y2);
\node[main] (x1) [below right = 0.25cm and 0.3cm of x3]{$X_1$};
\draw[->] (x1) to[bend right=20] (y2);
\node[main] (x2) [right of=h3]{$X_2$};
\draw[->, gray] (x2) -- (h3);
\draw[->] (x1) -- (x2);
\node[main] (x5) [above = 0.5cm of h1]{$X_5$};
\draw[->] (x5) -- (y1);
\draw[->] (x1) -- (x2);
\draw[->] (x1) to[bend right=30] (x5);
\node[main] [left of=y1](x7) {$X_7$}; 
\node[main] [above = 0.5cm of x7](x6) {$X_6$};
\draw[->] (y1) -- (x7);
\draw[->] (y1) -- (x6);
\node[main] [left of=y2](x8) {$X_8$};
\draw[->] (y2) -- (x8);
\node[main] [below = 0.5cm of h3](x4) {$X_4$};
\draw[->, gray] (h3) -- (x4);
\draw[->] (x2) -- (x4);
\end{tikzpicture}\\
\centering (a)
\end{minipage}\hfill
\begin{minipage}{0.45\textwidth}
\centering 
\begin{tikzpicture}[node distance={15mm}, thick, main/.style = {}] 
\node[main, text=red] (y1) {$Y_1$}; 
\node[main, text=red] (y2) [below of=y1]{$Y_2$}; 
\draw[<->, red] (y1) -- (y2);
\node[main] (x3) [right of=y1]{$X_3$};
\draw[->] (x3) -- (y1);
\node[main] (x1) [below right = 0.25cm and 0.3cm of x3]{$X_1$};
\draw[->] (x1) to[bend right=20] (y2);
\node[main] (x2) [right of=y2]{$X_2$};
\node[main]  [below = 0.5cm of x2](x4) {$X_4$};
\draw[<->] (x4) -- (y2);
\draw[->] (x2) -- (x4);
\draw[->] (x2) -- (y2);
\draw[->] (x1) -- (x2);
\node[main] (x5) [above =0.5cm of x3]{$X_5$};
\draw[->] (x5) -- (y1);
\draw[->] (x1) to[bend right=30] (x5);
\node[main] [left of=y1](x7) {$X_7$}; 
\node[main] [above = 0.5cm of x7](x6) {$X_6$};
\draw[->] (y1) -- (x6);
\draw[->] (y1) -- (x7);
\node[main] [left of=y2](x8) {$X_8$};
\draw[->] (y2) -- (x8);
\end{tikzpicture}\\
\centering (b)
\end{minipage}
\caption{(a) Generating DAG with latent transition ans source nodes; (b) induced mixed graph.}
\label{Transition and source node}
\end{figure}

We provide a numerical example to demonstrate that  the causal parents of a  target variable  in an induced mixed graph are effectively identified when applying the result of Theorem \ref{Teorema 1} to test for invariance.
Let us consider the data generating model  for a bivariate Gaussian target $Y=\{Y_1,Y_2\}$ represented by the causal DAG  in Figure \ref{Transition and source node}(a) including eight observed variables $X=\{X_1,\dots, X_8\}$ and  three latent variables $H=\{H_1,H_2,H_3\}$. Specifically, $H_1$ is a latent transition parent for the outcome $Y_1$; $H_3$ acts simultaneously as a latent transition parent for $Y_2$ and a source node for $X_4$ and $Y_2$; $H_2$ is a common source node for the two outcomes. The induced mixed graph over the observed variables, shown in Figure \ref{Transition and source node}(b), includes the bi-directed edges $Y_1 \leftrightarrow Y_2$ and  $X_4 \leftrightarrow Y_2$ that are non-causal relationships, and directed edges $X_3 \rightarrow Y_1$, $X_2 \rightarrow Y_2$, so that $X_2,X_3$ enter the augmented causal parent set of $Y_1,Y_2$. 
Section \ref{app:numerical1} of the Appendix provides a detailed description of the structural causal model and the generation of data from two environments, which mimic multiplicative interventions. In each environment, data were simulated with a sample size of $n=500$.
The  parent set $\{X_1, X_2,X_3, X_5\}$ of the joint outcome $\{Y_1, Y_2\}$ in Figure  \ref{Transition and source node}(b) is correctly identified $92\%$ of times across $50$ Monte Carlo replications using a 5\% significance level for the tests.

\section{Structural causal models with hidden confounders}\label{sec.SCMH}
\subsection{Mixed graph model for the observed variables}
 We frame the proposed procedure within the setting of causal models with hidden confounders, which typically make causal effects non-identifiable. Given the variable structure described in Section \ref{sec.SCM}, we further  assume that a set of instrumental variables $Z \in \mathbb{R}^r$  are available to ensure model identifiability in correspondence of confounding effect \citep{angrist1996}. 
Based on a generating DAG, we consider the  structural causal model  for $Y$, $H_{PA}, X_{PA}$
\begin{equation}\label{structural causal model}
Y = f(X_{PA}) + g(H_{PA}, \delta_{Y}), \qquad H_{PA} = \delta_H, \qquad X_{PA} = l(H_{PA}, Z, \eta) 
\end{equation}
with $\delta_Y \ci \{X_{PA}, H_{PA} \}$, $\delta_H \ci \delta_Y$, $\eta \ci \{H, Z \}$, 
 $Y \ci Z \vert    X_{PA}$,  $\mathbb{C}ov(Z, X_{PA}) \not = 0$
and 
$\delta_Y \sim F_{\delta_Y}$ for some
distribution $ F_{\delta_Y}$. 
As  interventions do not affect  $\{H_{PA},Y\}$, the distribution of $Y|\{X_{PA}, H_{PA}\}$ is invariant. 
The set $H_{PA}$ may involve  confounding. All latent parents are assumed to be exogenous ($X_{PA(H)}= \emptyset)$, consequently 
there are no latent parents acting as transition nodes for $Y$.
If we marginalize over $H_{PA}$, we get a mixed graph $\mathcal{G}$  with  additional bi-directed edges linking $Y_j$ and its parents $X_i \in X_{PA_j}$ in correspondence of confounding effect. Figures \ref{fig:DAG-latent}(d) and (h)  are illustrative  of for the case of  a single target. The mixed graph includes multiple edges to distinguish between the direct effect of parent variables on a target and the further confounding effect. The induced model for $Y$ given the observed variables is the structural equation model 
\begin{equation}\label{structural causal model with HC}
Y =f(X_{PA}) +  \varepsilon_{Y}, \quad X_{PA}= \tilde{l}(Z, \eta), \quad \varepsilon_{Y} \nci X_{PA}, 
\end{equation}
with an error term $\varepsilon_Y \sim F_{\varepsilon_Y}$.\
In this setting, we prove that in the  induced mixed graph,  parents of $Y$ are causal parents 
under the weak invariance principle,  and that the structural equation in \eqref{structural causal model with HC} is a  causal model for a target $Y$. 
We recall the definition of weak invariance \citep{peters}. 
\begin{definition}\label{Definition 2}
If the multivariate random variable $Y$ follows the causal equation $Y^e = f(X_{PA}^e, \varepsilon_Y^e)$ for all  $e \in \mathcal{E}$, then we say that the distribution of $Y\vert X_{PA}$ satisfies the weak invariance property if conditions (i) and (ii) of Definition \ref{Invariance definition} hold.
\end{definition} 
Intuitively, the weak invariance condition represents a softer form of regularity that holds even if hidden confounders are present. We  exploit this definition in order to give a causal interpretation to marginal mixed graph models also when this type of latent configuration occurs.

\begin{proposition}\label{Proposition 2}
 The marginal model in \eqref{structural causal model with HC}  is a  structural causal model, then the conditional distribution of $Y \vert X_{PA}$ is weakly invariant. 
\end{proposition}

Note that the induced mixed graph for a DAG including confounding effects may be not equal to the induced ancestral graph \citep{Zhang}. In fact, given the  DAG in Figure \ref{fig:DAG-latent}(d), the induced ancestral graph does not have the bi-directed edge between $X$ and $Y$. Instead, it has  a  directed arrow from $Z$ to $Y$, making $Z$ an additional parent of the target variable. Even though the ancestral and mixed graph models are equivalent in terms of independences,  the next proposition shows that the ancestral graph model is not invariant for $Y$. The proof in Section \ref{appendix A} of the Appendix is complemented with a brief discussion on the comparison between the induced mixed and ancestral graph model.
\begin{proposition}\label{prop:ancestral}
Consider the causal DAG in Figure \ref{fig:DAG-latent}(d). In the induced ancestral graph obtained by ignoring the latent variable $H$,  
the weak invariance property does not hold  for $p(Y \vert X, Z)$.
\end{proposition}

In a more general setting where the latent parents can also act as transition nodes, we prove that causal invariance holds when the effect of  $H_{PA}$ on $Y$ is linear. Then, we consider the   model
\begin{equation}\label{structural causal
model version 2}
Y = f(X_{PA}) + \Gamma_{PA}^TH + \delta_{Y}, \quad H_{PA} = g(X_{PA(H)}) + \delta_H, \quad X_{PA} = l(H, Z, \eta)
\end{equation}
with  $\delta_Y \ci \{X_{PA}, H_{PA} \}$, $\delta_H \ci X_{PA(H)}$, $\delta_H \ci \delta_Y$, $\eta \ci \{H, Z \}$ and $\delta_Y \sim F_{\delta_Y}$.    $\Gamma_{PA} \in \mathbb{R}^{q \times m}$ is a coefficient matrix with $(\Gamma_{PA})_{lj}=0$ if and only if $H_l \not \in H_{PA_j}$. 
The induced model for $Y$ is 
\begin{equation}\label{structural causal model with HC version 2}
Y =f(X_{\widetilde{PA}}) +  \varepsilon_{Y}, \quad X_{\widetilde{PA}}= \tilde{l}(Z, \eta), \quad \varepsilon_{Y} \nci X_{\widetilde{PA}},
\end{equation}
where $\varepsilon_Y \sim F_{\varepsilon_Y}$ and the  parent set $X_{\widetilde{PA}}$ can be larger than $X_{PA}$ in case of latent transition parents.
\begin{proposition}\label{prop:endogenousHC}
The marginal model in \eqref{structural causal model with HC version 2}  is a  structural causal model, then the conditional distribution of $Y \vert X_{\widetilde{PA}}$ is weakly invariant.
\end{proposition}
Given Proposition \ref{prop:endogenousHC}, invariance conditions are now derived for the case of a Gaussian target.

\subsection{Gaussian target variable}

We now specify the  structural causal model in \eqref{structural causal model version 2}
for a Gaussian target 
\begin{equation}\label{linear SCM}
Y= \Lambda_{PA}^TX + \Gamma_{PA}^TH + \delta_Y, \quad H_{PA} = \Phi_{PA}^TX + \delta_H,
\end{equation}
with $\delta_Y \sim N(0, \Omega)$ for a diagonal covariance matrix $\Omega$ and where the non-zero entries in the matrices $\Lambda_{PA}, \Gamma_{PA}, \Phi_{PA}$ identify the  parent sets $X_{PA}$, $H_{PA}$ and $X_{PA(H)}$.  
Given $B_{PA}^T = \Lambda_{PA}^T + \Gamma_{PA}^T\Phi_{PA}^T$ and $\varepsilon_Y = \Gamma_{PA}^T\delta_H + \delta_Y$,  the  causal equation for $Y$ ignoring $H_{PA}$ becomes
\begin{equation}\label{linear SCM with  HC}
    Y = B_{PA}^TX + \varepsilon_Y, \qquad \varepsilon_Y \sim N(0, \Sigma)
\end{equation}
where $\varepsilon_Y \nci X_{\widetilde{PA}}$ if any hidden confounder is present. The next theorem provides a necessary and sufficient condition for testing weak causal invariance for any subset $X_S \subseteq X$ in \eqref{linear SCM with  HC}.

\begin{theorem}\label{Teorema2}
In the above linear setting, let $S_j \subset \{1,\dots,p\}$ be a set of potential causal parents for $Y_j \in Y$ with $S=\bigcup_{j=1}^{m}S_j$. Let $B \in \mathbb{R}^{p \times m}$ such that $B_{ij}=0$ if and only if $X_i \not \in S_j$ and that solves the expected likelihood score equations
\begin{equation*} 
\mathbb{E}_{{X}_{S }, {Y}}[({Y} - B^T {X}) {X}_S^T]=  \Gamma_{PA}^T\mathbb{E}_{{X}_{S }, {\delta}_{H}}[\delta_H {X}_S^T], \; \;\;
\mathbb{E}_{{X}_{S }, {Y}, \delta_H}[({Y} - B^T {X}) \delta_H^T]= \Gamma_{PA}^T\mathbb{E}_{\delta_H}[\delta_H\delta_H^T].
\end{equation*}
Then, $S=\widetilde{PA}$ and $B=B_{PA}$ almost surely if and only if $
\mathbb{E}_{{X}, {Y} }[({Y} - B^T {X} )({Y} - B^T {X} )^T] = \Sigma .
$
\end{theorem}

\begin{figure}
\begin{minipage}{0.45\textwidth}
\centering 
\begin{tikzpicture}[node distance={15mm}, thick, main/.style = {}] 
\node[main, text=red] (y1) {$Y_1$}; 
\node[main, text=red] (y2) [below of=y1]{$Y_2$}; 
\node[main, text=gray] (h2) [below left = 0.1cm and 0.3cm of y1] {$H_2$};
\draw[->, gray] (h2) -- (y1);
\draw[->, gray] (h2) -- (y2);
\node[main, text=gray] (h1) [right of=y1]{$H_1$};
\draw[->, gray] (h1) -- (y1);
\node[main] (x3) [right of=h1]{$X_3$};
\draw[->, gray] (x3) -- (h1);
\node[main, text=gray] (h3) [below of=h1]{$H_3$};
\node[main] (x2) [right of=h3]{$X_2$};
\draw[->, gray] (h3) -- (y2);
\node[main] (x1) [above = 0.25cm of x2]{$X_1$};
\draw[->] (x1) to[bend right=2] (y2);
\node[main, text=blue] (z1) [right = 0.25cm of x1]{$Z_1$}; 
\draw[->, blue] (z1) -- (x1);
\draw[->, gray] (x2) -- (h3);
\draw[->] (x1) -- (x2);
\node[main] (x5) [above = 0.5cm of h1]{$X_5$};
\draw[->] (x5) -- (y1);
\draw[->] (x1) -- (x2);
\node[main, text=blue] (z5) [above = 0.5cm of x3]{$Z_5$}; 
\draw[->, blue] (z5) -- (x5);
\node[main, text=gray] (h4) [above = 0.5cm of y1]{$H_4$};
\draw[->, gray] (h4) -- (x5);
\draw[->, gray] (h4) -- (y1);
\node[main] [left of=y1](x7) {$X_7$}; 
\node[main] [above = 0.5cm of x7](x6) {$X_6$};
\draw[->] (y1) -- (x7);
\draw[->] (y1) -- (x6);
\node[main] [left of=y2](x8) {$X_8$};
\draw[->] (y2) -- (x8);
\node[main] [below = 0.5cm of h3](x4) {$X_4$};
\draw[->, gray] (h3) -- (x4);
\draw[->] (x2) -- (x4);
\node[main, text= gray] [above = 0.25cm of h3](h5) {$H_5$};
\draw[->, gray] (h5) -- (x1);
\draw[->, gray] (h5) -- (y2);
\end{tikzpicture}\\
\centering (a)
\end{minipage}\hfill
\begin{minipage}{0.45\textwidth}
\centering 
\begin{tikzpicture}[node distance={15mm}, thick, main/.style = {}] 
\node[main, text=red] (y1) {$Y_1$}; 
\node[main, text=red] (y2) [below of=y1]{$Y_2$}; 
\node[main] (x2) [right of=y2]{$X_2$};
\draw[<->, red] (y1) -- (y2);
\node[main] (x3) [right of=y1]{$X_3$};
\draw[->] (x3) -- (y1);
\node[main] [below right = 0.5cm and 0.1cm of y2](x4) {$X_4$};
\draw[<->] (x4) -- (y2);
\node[main] (x1) [above = 0.25cm of x2]{$X_1$};
\node[main, text=blue] (z1) [right = 0.25cm of x1]{$Z_1$}; 
\draw[->, blue] (z1) -- (x1);
\draw[->] (x1) -- (y2);
\draw[<->] (x1) to[bend right=20] (y2);
\draw[->] (x2) -- (x4);
\draw[->] (x2) -- (y2);
\draw[->] (x1) -- (x2);
\node[main] (x5) [above =0.5cm of x3]{$X_5$};
\draw[->] (x5) -- (y1);
\node[main, text=blue] (z5) [right = 0.25cm of x5]{$Z_5$}; 
\draw[->, blue] (z5) -- (x5);
\draw[<->] (x5) to[bend right=20] (y1);
\node[main] [left of=y1](x7) {$X_7$}; 
\node[main] [above = 0.5cm of x7](x6) {$X_6$};
\draw[->] (y1) -- (x6);
\draw[->] (y1) -- (x7);
\node[main] [left of=y2](x8) {$X_8$};
\draw[->] (y2) -- (x8);
\end{tikzpicture}\\
\centering (b)
\end{minipage}
\caption{ (a) Generating DAG   with different  latent variable configurations; (b) induced mixed graph.}
\label{Final example}
\end{figure}
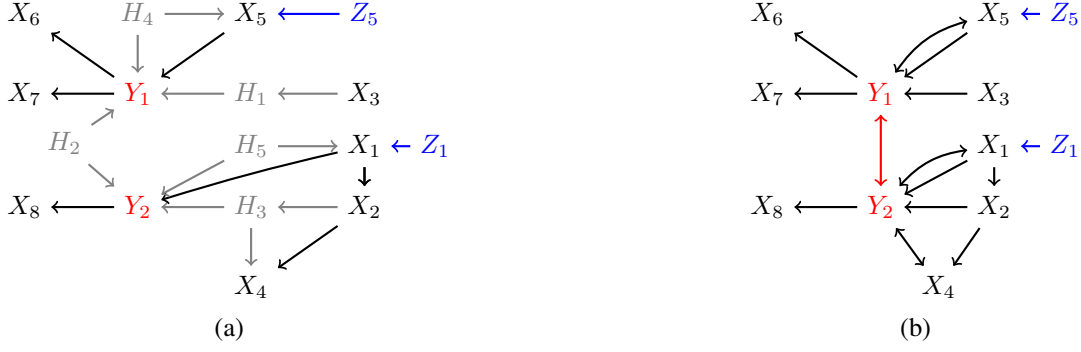

We provide a numerical illustration to show that by applying the result of Theorem \ref{Teorema2} to test for invariance in mixed graph models with confounding effects, the causal parents of a target variable  are effectively identified.
To this end, we consider the data generating model represented by the causal DAG in Figure   \ref{Final example}(a). The model includes a bivariate Gaussian target  $Y=\{Y_1,Y_2\}$, eight observed variables $X=\{X_1,\dots, X_8\}$ and  five latent variables $\{H_1,H_2,H_3,H_4, H_5\}$ serving as causal latent parents of $Y$. In particular,  the latent parents $H_4$ and $H_5$ act as hidden counfounders and the instrumental variables $Z_1$ and $Z_5$ are introduced to ensure the identifiability of the causal effects of $X_1$ on $Y_2$ and of $X_5$ on $Y_1$, respectively. The induced mixed graph over the observed variables, shown in Figure \ref{Final example}(b), includes additional bi-directed edges with no causal interpretation and the directed edges $X_3 \rightarrow Y_1$ and $X_2 \rightarrow Y_2$, so that $X_2$ and $X_3$ enter the augmented causal parent set of $Y_1,Y_2$. Data are sampled based on a data-generating model and from two environment mimic multiplicative intervention; for a detailed description  see Section \ref{app:numerical2} of the Appendix. 
Using the testing procedure from Theorem \ref{Teorema2} within the Monte Carlo simulation scheme described in Section \ref{sec.SCM}, the  parent set $\{X_1, X_2, X_3, X_5\}$ in Figure \ref{Final example}(b) is correctly identified in $92\%$ of the cases across $50$ replicates when $n=1000$ and using a 5\% significance level for the tests. 
On the other hand, the testing procedure based on the result of Theorem \ref{Teorema 1} identifies the correct causal parents only in  $18\%$ of the cases.

\section{Empirical illustration on flow cytometry data}
\label{app:application}
We consider the flow cytometry data of \cite{sachs2005}. These data measure the abundance of 11 biochemical agents in different environments. The first of them can be considered as an observational environment, while the others are generated under the intervention of a certain reagent that modifies the abundance or the action of the biochemical agents. 

We focus on the first observational environment and the one that corresponds to an abundance intervention on the agent PKC (due to the reagent PMA) in combination with a global intervention (see \cite{Mooji} and \cite{meinshausen2016methods} for further details). We concentrate on the bivariate target \{P38, Erk\}, as these nodes are at the same level according to the consensus network in \citet[][Fig. 2]{sachs2005}. 
Firstly, we search for the causal parents of the bivariate target among all the possible subsets of the 9 remaining covariates, namely Raf, Mek, Plcg, PIP2, PIP3, Akt, PKA, PKC and Jnk, without considering the presence of confounders in the causal effects. In other words, we apply the procedure resulting from Theorem \ref{Teorema 1} for parent identification. The set $\{\text{Jnk}, \text{PKC}, \text{Akt}\}$ is found as the parent set, with a p-value of $0.60$ from the Box's M test, and the causal model is represented by the graph  in Figure \ref{Application}a.

Suppose now that the variable Jnk is unobserved. As expected, the restricted set \{PKC, Akt\} no longer satisfies the invariance condition  required by  Theorem \ref{Teorema 1}. Given that  the consensus network shows a relationship between PKC and Jnk, it seems reasonable to assume that  Jnk  acts as a confounder of the effect of PKC  on P38.  
We therefore apply the procedure for parent identification  provided by Theorem \ref{Teorema2}, and employ PIP2, Plcg, Raf and Mek  as instrumental variables to ensure the identifiability of the causal effect of PKC on P38. This choice is further justified  by the consensus network, where these variables appear as  the most predictive of PKC.
In this second analysis, the invariance condition is  satisfied for the pair \{PKC, Akt\} with $p$-value equal to 0.55 and the structural causal model is represented by the mixed graph  in Figure \ref{Application}(b).

\begin{figure}[t]
\begin{minipage}{0.45\textwidth}
\centering 
\begin{tikzpicture}[node distance={15mm}, thick, main/.style = {}] 
\node[main, text=red] (p38) {P38}; 
\node[main, text=red] (erk) [below of=p38]{Erk}; 
\draw[<->, red] (p38) -- (erk);
\node[main] (pkc) [right = 1.5cm of p38]{PKC};
\node[main] (jnk) [below = 0.25cm  of pkc]{Jnk};
\draw[->] (jnk) -- (p38);
\draw[->] (pkc) -- (p38);
\node[main] (akt) [below of=pkc]{Akt};
\draw[->] (akt) -- (erk);
\draw[->] (pkc) -- (erk);
\draw[->] (jnk) -- (erk);
\draw[->] (akt) -- (p38);
\draw[->] (jnk) -- (p38);
\end{tikzpicture}\\
\centering (a)
\end{minipage}\hfill
\begin{minipage}{0.45\textwidth}
\centering
\begin{tikzpicture}[node distance={15mm}, thick, main/.style = {}] 
\node[main, text=red] (p38) {P38}; 
\node[main, text=red] (erk) [below of=p38]{Erk}; 
\draw[<->, red] (p38) -- (erk);
\node[main] (pkc) [right = 1.5cm of p38]{PKC};
\draw[->] (pkc) -- (p38);
\draw[<->] (pkc) to[bend right=20] (p38);
\node[main] (akt) [below of=pkc]{Akt};
\draw[->] (akt) -- (erk);
\draw[->] (akt) -- (p38);
\draw[->] (pkc) -- (erk);
\end{tikzpicture}\\
\centering (b)
\end{minipage}
\caption{Causal graphs found by the proposed procedure without (a) and with (b) confounding.}
\label{Application}
\end{figure}
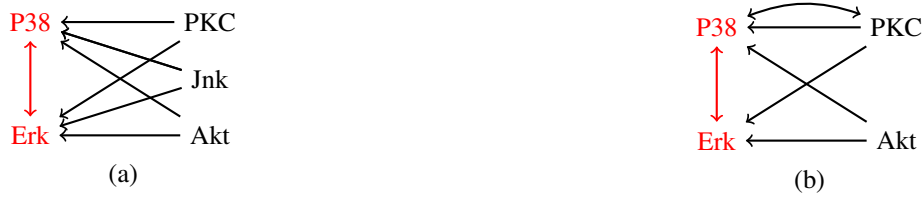

\section{Conclusion}

 Testing causal invariance in the absence of causal sufficiency, and in particular under latent confounding, is challenging. Rather than imposing probabilistic assumptions on the latent structure, in this paper we approach this issue by working directly with the marginal distributions of the observed variables. By considering the structural model of the target variable in the induced acyclic directed mixed graph,   we  demonstrate  that causal invariance is preserved for the observed parent set under specific latent configurations. Then, we derive  the formal conditions for testing such invariance in the case of a multivariate Gaussian target.
 These results extend both the principle and the application of causal invariance to a broader class of models where hidden confounding is present. An interesting future research question is to investigate under what conditions testing  causal invariance is feasible for the full set of observed variables rather than for a specific target.

\section*{Acknowledgments}
Financial support was provided to the first and second author by the
MUR-PRIN grant 2022 SMNNKY, CUP B53D23009470006, the MUR Department of Excellence project 2023-2027 ReDS 'Rethinking Data Science' - Department of Statistics, Computer Science, Applications - University of Florence. 
Veronica Vinciotti acknowledges funding from the the European Union - Next Generation EU, Mission 4 Component 2 - CUP C53D23002580006 (MUR-PRIN grant 2022SMNNKY).

\appendix

\section*{Appendix}

\section{Proofs of the technical  results in Sections \ref{sec.SCM} and \ref{sec.SCMH}}\label{appendix A}

\subsection*{Proof of Proposition \ref{Proposition 1}}
\begin{proof}
We assume that data are collected from different environments and that the structural causal model \eqref{eq:SCM-transition-source} holds for every $e \in \mathcal{E}$.
We can rewrite the first equation as $Y = f(X_{PA}, H_{PA}, \delta_Y) = f(X_{PA}, g(X_{PA(H)}, \delta_H), \delta_Y)$.
Moreover, 
\[
\mathbb{E}_{\delta_H}[Y] = \int f(X_{PA}, g(X_{PA(H)}, \delta_H), \delta_Y) p(\delta_H) d \delta_H =: \tilde{f}(X_{\widetilde{PA}}, \delta_Y)
\]
where (i) $\tilde{f}$ is the same for all the environments because it  depends only on $f,g$, which are the same $\forall e \in \mathcal{E}$ thanks to the invariance of $Y \vert \{X_{PA}, H_{PA} \}$ and $H_{PA} \vert X_{PA(H)}$, and on the distribution of $\delta_H$, which does not change across the environments since no external interventions affect $H_{PA}$, and finally since $X_{PA} \ci \delta_H$. Moreover, (ii) $\delta_Y \ci X_{\widetilde{PA}}$ and (iii) the distribution of $\delta_Y$ is the same in all the environments for hypothesis.
\end{proof}

\subsection*{Proof of Theorem \ref{Teorema 1}}
\begin{proof}
The first implication is immediate because if $S=\widetilde{PA}$ and $B=B_{PA}$ with probability 1, then
\[
\mathbb{E}_{{X},{Y}}[({Y} - B^T{X})({Y} - B^T{X})^T] =
\mathbb{E}_{{X},{Y}}[({Y} - B_{PA}^T{X})({Y} - B_{PA}^T{X})^T] = \Sigma.
\]
We focus now on the other direction. First of all, the population log-likelihood function is $l_{{X},{Y}}(B) = -\frac{1}{2} \cdot \left( {Y}-B^T{X}\right)^T \Sigma^{-1} \left( {Y}-B^T{X}\right).$
Let ${b}:=vec(B)=(\beta_1, \dots, \beta_m)^T \in \mathbb{R}^{pm}$.
The column vector $B^T{X}$  can be written as $B^T{X}= vec(B^T {X})=vec({X}^TB)=(I_m \otimes {X})^Tvec(B)= (I_m \otimes {X})^T{b}$,
where $I_m$ is an identity matrix of size $m$. Following \citet[Thr.3]{alice}, the population log-likelihood function can be expressed as $l_{{X},{Y}}({b}) = -\frac{1}{2} \cdot \left( {Y}-(I_m \otimes {X})^T{b}\right)^T \Sigma^{-1} \left( {Y}-(I_m \otimes {X})^T{b}\right)$.
The first derivative with respect to ${b}$ is equal to $\frac{\partial l_{{X},{Y}}({b})}{\partial {b}} = (I_m \otimes {X}_S) \Sigma^{-1} \left( {Y}-(I_m \otimes {X})^T{b}\right)$
while the first derivative with respect to $B^T$ is equal to $\frac{\partial l_{{X},{Y}}({b})}{\partial {b}^T} = \left( {Y}-(I_m \otimes {X})^T{b}\right)^T \Sigma^{-1} (I_m \otimes {X}_S^T)$.
The second derivative of the population log-likelihood function is
$\frac{\partial^2 l_{{X},{Y}}({b})}{\partial {b} \partial {b}^T} = 
-(I_m \otimes {X}_S) \Sigma^{-1}(I_m \otimes {X}_S^T)$. 
From the invariance condition, we know that
$\mathbb{E}_{ {X}, {Y}}[( {Y} - B^T {X})( {Y} - B^T {X})^T] = \Sigma$.
The first term can be written as
\[
\mathbb{E}_{ {X}, {Y}}[( {Y} - B^T {X})( {Y} - B^T {X})^T] = \mathbb{E}_{ {X}_S}\left[ \mathbb{E}_{ {Y}}[( {Y} - B^T {X})( {Y} - B^T {X})^T \vert  {X}_S] \right]
\]
and hence $\mathbb{E}_{ {Y}}[( {Y} - B^T {X})( {Y} - B^T {X})^T \vert  {X}_S] = \Sigma$
holds almost surely on the space of distributions on $ {X}_S$. 
Multiplying  by $\Sigma^{-1}$ on both sides, we obtain
$
\mathbb{E}_{ {Y}}[\Sigma^{-1}( {Y} - B^T {X})( {Y} - B^T {X})^T \Sigma^{-1}\vert  {X}_S] = \Sigma^{-1}$.
Multiplying now by the Kronecker product on the left by $(I_m \otimes  {X}_S)$ and on the right by $(I_m \otimes  {X}_S^T)$, we get
\[
\mathbb{E}_{ {Y}}[(I_m \otimes  {X}_S)\Sigma^{-1}( {Y} - (I_m \otimes  {X}_S) {b})( {Y} - (I_m \otimes  {X}_S) {b})^T \Sigma^{-1}(I_m \otimes  {X}_S^T)\vert  {X}_S]  =(I_m \otimes  {X}_S) \Sigma^{-1}(I_m \otimes  {X}_S^T).
\]
Finally, by taking the expectation with respect to $ {X}_S$,  
$$
\mathbb{E}_{ {X}_S,  {Y}}\left[\frac{\partial l_{ {X}, {Y}}( {b})}{\partial  {b}} \cdot \frac{\partial l_{ {X}, {Y}}( {b})}{\partial  {b}^T}\right]
=
- \mathbb{E}_{ {X}_S}\left[\frac{\partial^2 l_{ {X}, {Y}}( {b})}{\partial  {b} \partial  {b}^T} \right].
$$
If we prove that $\mathbb{E}_{ {X}_S,  {Y}}\left[\frac{\partial l_{ {X}, {Y}}( {b})}{\partial  {b}} \right] =0,$
then the Bartlett's identity holds and thus $
 {Y} \; \vert  {X} \sim N(B^T {X},\Sigma)$.
This follows from the hypothesis $\mathbb{E}_{ {X}_{S}, {Y}}[( {Y} - B^T {X}) {X}_{S}^T]=0$. Indeed, since $\Sigma$ is symmetric and positive definite, if we multiply on the left by $\Sigma^{-1}$ both sides we obtain $\mathbb{E}_{ {X}_{S}, {Y}}[ \Sigma^{-1}( {Y} - B^T {X}) {X}_{S}^T]=0$.
Moreover,
\begin{align*}
& \mathbb{E}_{ {X}_{S}, {Y}}[ \Sigma^{-1}( {Y} - B^T {X}) {X}_{S}^T]=0 \Rightarrow \mathbb{E}_{ {X}_{S}, {Y}}[ {X}_{S} ( {Y} - B^T {X})^T \Sigma^{-1}]=0 \\
&\Rightarrow \mathbb{E}_{ {X}_{S}, {Y}}\left[ vec\left( {X}_{S} ( {Y} - B^T {X})^T \Sigma^{-1} \right) \right]=0 \Rightarrow \mathbb{E}_{ {X}_{S}, {Y}}\left[ (I_m \otimes  {X}_{S}) vec\left( ( {Y} - B^T {X})^T \Sigma^{-1} \right) \right]=0 \\
&\Rightarrow \mathbb{E}_{ {X}_{S}, {Y}}\left[ (I_m \otimes  {X}_{S}) vec\left(  \Sigma^{-1}( {Y} - B^T {X}) \right) \right]=0 \Rightarrow \mathbb{E}_{ {X}_{S}, {Y}}\left[ (I_m \otimes  {X}_{S})  \Sigma^{-1}( {Y} - B^T {X}) \right]=0\\
&\Rightarrow \mathbb{E}_{ {X}_S,  {Y}}\left[\frac{\partial l_{ {X}, {Y}}( {b})}{\partial  {b}} \right] =0.
\end{align*}
Since the set $\{B \in \mathbb{R}^{p \times m} : \mathbb{E}_{ {X}, {Y}}[( {Y} - B^T {X})( {Y} - B^T {X})^T] = \Sigma \}$ is a surface of dimension $pm-m$ in the space $\mathbb{R}^{pm}$, its Lebesgue measure in $\mathbb{R}^{pm}$ is zero. Therefore, $S=\widetilde{PA}$ and $B=B_{PA}$ almost surely.
\end{proof}

\subsection*{Proof of Proposition \ref{Proposition 2}}
\begin{proof}
We consider data coming from multiple environments and we assume that for all $e \in \mathcal{E}$ they are generated from the structural causal model \eqref{structural causal model}. 
As $H_{PA}= \delta_H$ is exogenous, we rewrite the first equation as $Y= f(X_{PA}) + g(\delta_H, \delta_Y) = f(X_{PA}) + \varepsilon_Y$
by defining $\varepsilon_Y := g(\delta_H, \delta_Y)$. Then, (i) $f$ is the same across the environments thanks to the invariance of the distribution of $Y \vert \{ X_{PA}, H_{PA} \}$. Moreover, (ii) the distribution of $\varepsilon_Y$ is the same for all $e \in \mathcal{E}$. Indeed, the distribution of $\varepsilon_Y$ is completely determined by $g$, which is the same across the environments for hypothesis and by $\delta_H$, which does not change because the external interventions are not performed on $H_{PA}$. Moreover, the marginal distribution of $\varepsilon_Y$ over $\delta_H$ is the same for all $e \in \mathcal{E}$ because $\delta_Y \ci \delta_H$. 
\end{proof}

\subsection*{Proof of Proposition \ref{prop:ancestral}}
\begin{proof}
We show a counter example. We consider a generic linear structural causal model associated to the DAG in Figure \ref{fig:DAG-latent}(d) where $H$ is an unobserved counfounder between $X$ and $Y$
\begin{equation}\label{SCM 4 variables}
Z = \delta_Z, \quad  H = \delta_H, \quad X = \lambda_{xz}Z + \gamma_{xh}H + \delta_X, \quad
Y = \lambda_{yx}X + \gamma_{yh}H + \delta_Y
\end{equation}
where $(\delta_Z, \delta_H, \delta_X, \delta_Y) \sim N( {0}, \Omega)$ with $\Omega = diag(\omega_{zz}, \omega_{hh}, \omega_{xx}, \omega_{yy})$. We assume data are collected from two environments $e$ and $e'$, where $\omega_{xx}^{e} \not = \omega_{xx}^{e'}$, while the distributions of $H$ and of $Y \vert X, H$ remain the same. The structural equation for $Y$ in the induced ancestral graph model is
$$
Y = \beta_{yz \vert x}Z + \beta_{yx \vert z}X + \varepsilon_Y
$$
for some coefficients $\beta_{yz \vert x}$ and $\beta_{yx \vert z}$ and error term $\varepsilon_Y$. But $\beta_{yx \vert z}= \frac{\mathbb{C}ov(X,Y \vert Z)}{\mathbb{V}ar(X \vert Z)} = \lambda_{yx} + \frac{\gamma_{xh}\gamma_{yh}\omega_{hh}}{\gamma_{xh}^2\omega_{hh} + \omega_{xx}}$ and since we have assumed that the variance $\omega_{xx}$ is different in the two environments (while all the other quantities remain the same), then $\beta_{yx \vert z}^{e} \not = \beta_{yx \vert z}^{e'}$ so the invariance condition fails (as well as the weak invariance condition). It is easy to show that also the coefficient $\beta_{yz \vert x}$ and the variance of $\varepsilon_Y$ change in the two environments.
\end{proof}
For sake of completeness,  we discuss that if we consider the induced mixed graph for the counter example used in the proof of Proposition \ref{prop:ancestral}, the  model for $Y$ satisfies the invariance principle, more specifically, the conditional distribution of $Y \vert X$ is weakly invariant under the mixed graph model. Indeed, we assume data are collected from two environments $e,e'$ and for both scenarios data are generated according to the structural causal model (\ref{SCM 4 variables}). Then, $Y = \lambda_{yx}X + \gamma_{yh}\delta_H + \delta_Y = \lambda_{yx}X + \varepsilon_Y$, where $\varepsilon_Y= \gamma_{yh}\delta_H + \delta_Y$. The mean of $\varepsilon_Y$ is equal to zero in both $e$ and $e'$.   Its variance is $\mathbb{V}ar(\varepsilon_Y) = \gamma_{yh}^2\omega_{hh} + \omega_{yy}$.
Therefore, if we do not intervene on $Y$ and on $H$, the distribution of $\varepsilon_Y$ remains the same across the environments. Hence we can conclude that the distribution of $Y \vert X$ is weakly invariant. 

\subsection*{Proof of Proposition \ref{prop:endogenousHC}}
\begin{proof}
We assume that data are generated from different environments and that for each $e \in \mathcal{E}$ the causal model \eqref{structural causal model version 2} is satisfied.
We can write
\[
Y = f(X_{PA}) + \Gamma_{PA}^Tg(X_{PA(H)}) + \Gamma_{PA}^T \delta_H + \delta_Y = f(X_{PA}) + \Gamma_{PA}^Tg(X_{PA(H)}) + \varepsilon_Y
\]
where we have defined $\varepsilon_Y := \Gamma_{PA}^T\delta_H + \delta_Y$. Then, thanks to the hypothesis, we have that (i) $f(X_{PA}) + \Gamma_{PA}^Tg(X_{PA(H)})$ is the same for all $e \in \mathcal{E}$. Moreover, by repeating the same argument as in the proof of Proposition \ref{Proposition 2}, also (ii) the distribution of $\varepsilon_Y$ is the same across the environments as well as its marginal distribution with respect to $\delta_H$.
\end{proof}

\subsection*{Proof of Theorem \ref{Teorema2}}
\begin{proof}
The first implication is trivial, so we consider the other one. 
From the hypothesis we know that $\mathbb{E}_{ {X},  {Y} }[( {Y} - B^T  {X} ) ( {Y} - B^T  {X} )^T] = \Sigma =\Gamma_{PA}^T\mathbb{E}_{\delta_H}[\delta_H \delta_H^T] \Gamma_{PA} + \Omega.$
Therefore,
\[
\begin{split}
\mathbb{E}&_{ {X},  {Y}, \delta_H }[( {Y} - B^T  {X} -\Gamma_{PA}^T\delta_H) ( {Y} - B^T  {X} -\Gamma_{PA}^T\delta_H)^T]=\\
&= \mathbb{E}_{ {X},  {Y} }[( {Y} - B^T  {X} ) ( {Y} - B^T  {X})^T] -\mathbb{E}_{ {X},  {Y}, \delta_H }[( {Y} - B^T  {X} )\delta_H^T]\Gamma_{PA} 
-\Gamma_{PA}^T\mathbb{E}_{ {X},  {Y}, \delta_H }[\delta_H( {Y} - B^T  {X} )^T]+ \\
&+ \Gamma_{PA}^T\mathbb{E}_{\delta_H}[\delta_H\delta_H^T]\Gamma_{PA} = \Omega.
\end{split}
\]
We rewrite the first population likelihood score equation as 
$
\mathbb{E}_{ {X}_{S },  {Y}, \delta_H}[( {Y} - B^T  {X} - \Gamma_{PA}^T \delta_H)  {X}_S^T]=   {0}.
$
Then, by repeating the same argument as in the proof of Theorem 1 with the last two results we obtain that $ {Y} \vert  {X}, \delta_H \sim N\left(B^T  {X} + \Gamma_{PA}^T\delta_H, \Omega \right)$. 
Hence we get $ {Y} = B^T  {X} + \Gamma_{PA}^T\delta_H + \delta_Y = B^T  {X} + \varepsilon_{ {Y}}$ with probability one.
This thus implies that $S=\widetilde{PA}$ and $B=B_{PA}$ almost surely.
\end{proof}

\section{Numerical examples}\label{app:numerical}
This section describes the generating models of the numerical examples in Sections \ref{sec.SCM} and \ref{sec.SCMH}.
\subsection{Numerical example in Section \ref{sec.SCM}}\label{app:numerical1}
Given the causal  DAG  in Figure~\ref{Transition and source node}(a), data are generated from the following structural causal model
\begin{align*}
X_1^e &= \delta_{X_1}^e & X_2^e &= \lambda_{X_2,X_1} + \delta_{X_2}^e \\
X_3^e &= \delta_{X_3}^e & H_3^e &= \lambda_{H_3, X_2} X_2^e + \delta_{H_3}^e \\
X_4^e &= \lambda_{X_4,X_2}X_2^e + \gamma_{X_4, H_3}H_3^e + \delta_{X_4}^e & X_5^e &= \lambda_{X_5,X_1}X_1^e + \delta_{X_5}^e \\
H_1^e &= \lambda_{H_1,X_3}X_3^e + \delta_{H_1}^e & H_2^e &= \delta_{H_2}^e \\
X_6^e &= \lambda_{X_6, Y_1}Y_1^e + \delta_{X_6}^e & X_7^e &= \lambda_{X_7, Y_1}Y_1^e + \delta_{X_7}^e \\
X_8^e &= \lambda_{X_8, Y_2}Y_2^e + \delta_{X_8}^e & Y_1^e &= \lambda_{Y_1, X_5} X_5^e + \gamma_{Y_1,H_1}H_1^e + \gamma_{Y_1,H_2}H_2^e + \delta_{Y_1}^e \\
Y_2^e &= \lambda_{Y_2,X_1}X_1^e + \gamma_{Y_2,H_2}H_2^e + \gamma_{Y_2,H_3}H_3^e + \delta_{Y_2}^e
\end{align*}
where the error terms $\delta_{Y_i}^e, \delta_{X_j}^e, \delta_{H_r}^e \sim N(0,1)$ are all mutually independent and the regression coefficients are set to $\lambda_{X_2,X_1}=1$, $\lambda_{H_3, X_2}=2$,  $\lambda_{X_4,X_2}=-2.1$, $\gamma_{X_4, H_3}=-1.7$, $\lambda_{X_5,X_1} = 1.5$, $\lambda_{H_1,X_3}=2$, $\lambda_{Y_1, X_5}=-2.3$, $\gamma_{Y_1,H_1} = 1.5$, $\gamma_{Y_1,H_2} =-0.9$, $\lambda_{Y_2,X_1}=3.2$, $\gamma_{Y_2,H_2}=-1.3$, $\gamma_{Y_2,H_3}=2.1$, $\lambda_{X_6, Y_1}=3.3$, $\lambda_{X_7, Y_1} = 2.4$, $\lambda_{X_8, Y_2}=-1.6$.
Data are simulated from an observational environment ($e=1$) and from an interventional one $(e=2)$. 
We consider multiplicative interventions, i.e., we set $\delta_{X_j}^{e=2} = A_j\delta_{X_j}^{e=1}$ for some random variables $A_j \sim N(2,1)$ which are mutually independent and also independent of all the error terms. 
We remark that we work under the assumption that the target and the latent variables are not subject to external interventions.

The induced mixed graph obtained after marginalizing over the latent variables is represented in Figure~\ref{Transition and source node}(a) and, from Proposition \ref{Proposition 1}, the related causal equations for $Y_1$ and $Y_2$ become
\begin{equation*}\label{scm Y1 Y2}
Y_1^e = \beta_{Y_1,X_3}X_3^e + \beta_{Y_1, X_5}X_5^e + \varepsilon_{Y_1}^e, \qquad
Y_2^e = \beta_{Y_2,X_1}X_1^e + \beta_{Y_2, X_2}X_2^e + \varepsilon_{Y_2}^e
\end{equation*}
where $\beta_{Y_1,X_3}= \lambda_{H_1, X_3} \gamma_{Y_1, H_1} =3$, $\beta_{Y_1, X_5}= \lambda_{Y_1,X_5}=-2.3$,  $\beta_{Y_2,X_1}=\lambda_{Y_2,X_1}=3.2$, $\beta_{Y_2, X_2}= \lambda_{H_3, X_2} \gamma_{Y_2, H_3}= 4.2$, $\varepsilon_{Y_1}^e = \gamma_{Y_1,H_1}\delta_{H_1}^e +\gamma_{Y_1,H_2}\delta_{H_2}^e + \delta_{Y_1}^e$ and $\varepsilon_{Y_2}=\gamma_{Y_2,H_2}\delta_{H_2}^e + \gamma_{Y_2,H_3}\delta_{H_3}^e + \delta_{Y_2}^e$. Therefore, we have $\varepsilon_{Y_1}^e\sim N(0, 4.06)$, $\varepsilon_{Y_2}^e \sim N(0, 7.1)$ and $\mathbb{C}ov(\varepsilon_{Y_1} \varepsilon_{Y_2})= 1.17$.

\subsection{Numerical example in Section \ref{sec.SCMH}}\label{app:numerical2}
In this section we consider the generating DAG in Figure~\ref{Final example}(a) that includes different types of latent variables. The aim is to find the causal parents of the multivariate target $\{Y_1,Y_2\}$ by using the more general result in Theorem \ref{Teorema2}. This procedure needs instrumental variables in order to overcome the possibility of non-identifiability of the model. We assume to have prior knowledge that the causal effects from $X_1$ and from $X_5$ can be confounded.
We therefore consider the instrumental variables $Z_1$ and $Z_5$ for the predictors $X_1$ and $X_5$, respectively (see Figure \ref{Final example}). The associated causal model is
\begin{align*}
Z_1^e &= \delta_{Z_1}^e, \; Z_5^e = \delta_{Z_5}^e  &
H_r^e &= \delta_{H_r}^e, \text{ for } r \in \{2,4,5\} \\
X_1^e &= \lambda_{X_1,Z_1}Z_1^e + \gamma_{X_1,H_5}H_5^e + \delta_{X_1}^e & X_2^e &= \lambda_{X_2, X_1}X_1^e + \delta_{X_2}^e  \\
X_3^e &= \delta_{X_3}^e &
H_3^e &= \lambda_{H_3,X_2}X_2^e + \delta_{H_3}^e \\
X_4^e &= \lambda_{X_4,X_2}X_2^e + \gamma_{X_4,H_3}H_3^e + \delta_{X_4}^e & X_5^e &= \lambda_{X_5,Z_5}Z_5^e + \gamma_{X_5,H_4}H_4^e + \delta_{X_5}^e \\
H_1^e &= \lambda_{H_1,X_3}X_3^e + \delta_{H_1}^e &
X_6^e &= \lambda_{X_6, Y_1}Y_1^e  +  \delta_{X_6}^e \\
X_7^e &= \lambda_{X_7, Y_1}Y_1^e + \delta_{X_7}^e &
X_8^e &= \lambda_{X_8, Y_2}Y_2^e  + \delta_{X_8}^e \\
Y_1^e &= \lambda_{Y_1,X_5}X_5^e + \gamma_{Y_1,H_1}H_1^e + \gamma_{Y_1,H_2}H_2^e + \gamma_{Y_1,H_4}H_4^e + \delta_{Y_1}^e \\
Y_2^e &= \lambda_{Y_2,X_1}X_1^e + \gamma_{Y_2,H_2}H_2^e + \gamma_{Y_2,H_3}H_3^e + \gamma_{Y_2,H_5}H_5^e + \delta_{Y_2}^e  \\
\end{align*}
where $\lambda_{X_1, Z_1} = 1.5$, $\gamma_{X_1,H_5}=2.1$, $\lambda_{X_2, X_1}= 1$, $\lambda_{X_4,X_2}=-2.1$, $\gamma_{X_4,H_3} = -1.7$, 
$\lambda_{X_5,Z_5}=1.5$,
$\gamma_{X_5,H_4}=1.9$,
$\lambda_{H_1,X_3}=2$,
$\lambda_{H_3,X_2}=2$,
$\lambda_{Y_1,X_5}=-2.3$,
$\gamma_{Y_1,H_1}=1.5$,
$\gamma_{Y_1,H_2}=-0.9$,
$\gamma_{Y_1,H_4}=-1.7$,
$\lambda_{Y_2,X_1}=3.2$,
$\gamma_{Y_2,H_2}=-1.3$,
$\gamma_{Y_2,H_3} = 2.1$,
$\gamma_{Y_2, H_5}=2.6$,
$\lambda_{X_6, Y_1}=3.3$,
$\lambda_{X_7, Y_1}=2.4$,
$\lambda_{X_8, Y_2}=-1.6$.
As before, we consider multiplicative interventions. In particular, we set $\delta_{X_j}^{e=2} = A_j\delta_{X_j}^{e=1}$ for some random variables $A_j \sim N(2,2)$ which are mutually independent and also independent of all the error terms. 
The induced mixed graph obtained after marginalizing over the latent variables is shown in Figure \ref{Final example}(b) and, from Proposition \ref{prop:endogenousHC}, the causal equations for $Y_1$ and $Y_2$ become
\[
Y_1^e = \beta_{Y_1,X_3}X_3^e + \beta_{Y_1,X_5}X_5^e + \varepsilon_{Y_1}^e, \qquad Y_2^e = \beta_{Y_2,X_1}X_1^e + \beta_{Y_2,X_2}X_2^e + \varepsilon_{Y_2}^e
\]
where $\beta_{Y_1,X_3}= \gamma_{Y_1,H_1}\lambda_{H_1,X_3}=3$, $\beta_{Y_1,X_5}= \lambda_{Y_1,X_5}=-2.3$, $\beta_{Y_2,X_1}= \lambda_{Y_2,X_1}=3.2$, $\beta_{Y_2, X_2}=\gamma_{Y_2,H_3}\lambda_{H_3,X_2}=4.2$,
$\varepsilon_{Y_1}^e = \gamma_{Y_1,H_1}\delta_{H_1}^e + \gamma_{Y_1,H_2}\delta_{H_2}^e + \gamma_{Y_1,H_4}\delta_{H_4}^e + \delta_{Y_1}^e$ and 
$\varepsilon_{Y_2}^e=\gamma_{Y_2,H_2}\delta_{H_2}^e + \gamma_{Y_2,H_3}\delta_{H_3}^e + \gamma_{Y_2, H_5}\delta_{H_5}^e + \delta_{Y_2}^e$. Therefore, we have $\varepsilon_{Y_1}^e\sim N(0, 6.95)$, $\varepsilon_{Y_2}^e \sim N(0, 13.86)$ and $\mathbb{C}ov(\varepsilon_{Y_1} \varepsilon_{Y_2})= 1.17$. In particular, $\varepsilon_{Y_1} \nci X_5$ for the presence of the hidden confounder $H_4$ and
$\varepsilon_{Y_2} \nci X_1$ for the presence of the hidden confounder $H_5$.

\bibliographystyle{unsrtnat}
\bibliography{references}  

@article{henzi25,
    author = {Henzi, Alexander and Shen, Xinwei and Law, Michael and Bühlmann, Peter},
    title = {Invariant probabilistic prediction},
    journal = {Biometrika},
    volume = {112},
    issue = {1},
    pages = {doi:10.1093/biomet/asae063},
    year = {2025},
    doi = { 10.1093/biomet/asae063},
    url = {https://doi.org/10.1093/biomet/asae063}
 }

@Article{Zhang,
 Author = {Zhang, Jiji},
 Title = {Causal reasoning with ancestral graphs},
 Journal = {Journal of Machine Learning Research},
 ISSN = {1532-4435},
 Volume = {9},
 Pages = {1437--1474},
 Year = {2008},
 Language = {English},
 URL = {www.jmlr.org/papers/v9/zhang08a.html},
 zbMATH = {5968910},
 Zbl = {1225.68254}
}

@article{kook25,
author = {Lucas Kook and Sorawit Saengkyongam and Anton Rask Lundborg and Torsten Hothorn and Jonas Peters},
title = {Model-Based Causal Feature Selection for General Response Types},
journal = {Journal of the American Statistical Association},
volume = {120},
number = {550},
pages = {1090--1101},
year = {2025},
doi = {10.1080/01621459.2024.2395588}
}

@article{peters,
  title={Causal inference by using invariant prediction: identification and confidence intervals},
  author={Peters, Jonas and B{\"u}hlmann, Peter and Meinshausen, Nicolai},
  journal={Journal of the Royal Statistical Society Series B: Statistical Methodology},
  volume={78},
  number={5},
  pages={947--1012},
  year={2016},
  publisher={Oxford University Press}
}

@article{alice,
      title={Causal generalized linear models via {Pearson} risk invariance.}, 
      author={Alice Polinelli and Veronica Vinciotti and Ernst C. Wit},
      journal={Journal of Causal Inference},
volume = {},
pages = {},
year = {2026}
}

@article{rothenhausler2019causal,
  title={Causal {Dantzig}},
  author={Rothenh{\"a}usler, Dominik and B{\"u}hlmann, Peter and Meinshausen, Nicolai},
  journal={The Annals of Statistics},
  volume={47},
  number={3},
  pages={1688--1722},
  year={2019},
  publisher={JSTOR}
}

@book{pearl2009causality,
  title={Causality: Models, Reasoning and Inference},
  author={Pearl, Judea},
  year={2009},
  publisher={New York: Cambridge University Press}
}

@article{richardson2003,
  title={Markov properties for acyclic directed mixed graphs},
  author={Richardson, Thomas},
  journal={Scandinavian Journal of Statistics},
  volume={30},
  number={1},
  pages={145--157},
  year={2003},
  publisher={Wiley Online Library}
}

@article{angrist1996,
  title={Identification of causal effects using instrumental variables},
  author={Angrist, Joshua D and Imbens, Guido W and Rubin, Donald B},
  journal={Journal of the American statistical Association},
  volume={91},
  number={434},
  pages={444--455},
  year={1996},
  publisher={Taylor \& Francis}
}

@article{sachs2005,
  title={Causal protein-signaling networks derived from multiparameter single-cell data},
  author={Sachs, Karen and Perez, Omar and Pe'er, Dana and Lauffenburger, Douglas A and Nolan, Garry P},
  journal={Science},
  volume={308},
  number={5721},
  pages={523--529},
  year={2005},
  publisher={American Association for the Advancement of Science}
}

@article{meinshausen2016methods,
  title={Methods for causal inference from gene perturbation experiments and validation},
  author={Meinshausen, Nicolai and Hauser, Alain and Mooij, Joris M and Peters, Jonas and Versteeg, Philip and B{\"u}hlmann, Peter},
  journal={Proceedings of the National Academy of Sciences},
  volume={113},
  number={27},
  pages={7361--7368},
  year={2016},
  publisher={National Academy of Sciences}
}

@article{long2023estimating,
  title={Estimating causal effects with hidden confounding using instrumental variables and environments},
  author={Long, James P and Zhu, Hongxu and Do, Kim-Anh and Ha, Min Jin},
  journal={Electronic Journal of Statistics},
  volume={17},
  number={2},
  pages={2849},
  year={2023}
}

@article{rothenhausler2021anchor,
  title={Anchor regression: Heterogeneous data meet causality},
  author={Rothenh{\"a}usler, Dominik and Meinshausen, Nicolai and B{\"u}hlmann, Peter and Peters, Jonas},
  journal={Journal of the Royal Statistical Society Series B: Statistical Methodology},
  volume={83},
  number={2},
  pages={215--246},
  year={2021},
  publisher={Oxford University Press}
}

@article{Mooji,
author = {Mooij, Joris and Heskes, Tom},
year = {2013},
month = {09},
pages = {},
title = {Cyclic Causal Discovery from Continuous Equilibrium Data},
journal = {Uncertainty in Artificial Intelligence - Proceedings of the 29th Conference, UAI 2013}
}

@article{SadLau-2014,
author = {Sadeghi, K. and Lauritzen, S.},
  volume={20},
  number={2},
  pages={676--696},
year = {2014},
title = {Markov properties for mixed graphs},
journal = {Bernoulli}
}

@article{henckel-al-2024,
author = {Henckel, L. and Buttenschoen, M. and Maathuis, M.},
  volume={111},
  number={3},
  pages={771--788},
year = {2024},
title = {Graphical tools for selecting conditional instrumental sets},
journal = {Biometrika}
}

@article{gnecco2026boosted,
  title={Boosted Control Functions: Distribution generalization and invariance in confounded models},
  author={Gnecco, Nicola and Peters, Jonas and Engelke, Sebastian and Pfister, Niklas},
  journal={Journal of Machine Learning Research},
  volume={27},
  number={46},
  pages={1--57},
  year={2026},
  url={https://jmlr.org}
}

\end{document}